\documentclass[11pt]{amsart}

\pdfoutput = 1

\usepackage[english]{babel}
\usepackage[autostyle]{csquotes}
\usepackage{amsmath,amssymb,amsthm}
\usepackage{amsfonts}
\usepackage{amsxtra}

\usepackage{array,dcolumn}
\usepackage{graphicx}
\usepackage[hyperfootnotes=true,
        pdffitwindow=true,
        plainpages=false,
        pdfpagelabels=true,
        pdfpagemode=UseOutlines,
        pdfpagelayout=SinglePage,
                hyperindex,]{hyperref}

\usepackage{lmodern}
\usepackage[T1]{fontenc}
\usepackage[utf8]{inputenc}
\usepackage[activate={true,nocompatibility},spacing,kerning]{microtype}
\usepackage{bm}
\usepackage{bbm}
\usepackage[sc,osf]{mathpazo}
\microtypecontext{spacing=nonfrench}

  \calclayout

\newcommand{\mathsym}[1]{{}}
\newcommand{\unicode}[1]{{}}

\theoremstyle{plain}
\newtheorem{theorem}{Theorem}

\newtheorem{corollary}[theorem]{Corollary}
\newtheorem{proposition}[theorem]{Proposition}

\theoremstyle{definition}

\theoremstyle{remark}
\newtheorem{remark}[theorem]{Remark}

\numberwithin{equation}{section}

\numberwithin{theorem}{section}

\begin{document}


\title[Dip-ramp-plateau for the LUE structure function]{Quantifying dip-ramp-plateau for the Laguerre unitary ensemble structure function}
\author{Peter J. Forrester}
\address{School of Mathematics and Statistics, 
ARC Centre of Excellence for Mathematical \& Statistical Frontiers,
University of Melbourne, Victoria 3010, Australia}
\email{pjforr@unimelb.edu.au}

\date{\today}


\begin{abstract}
The ensemble average of $| \sum_{j=1}^N e^{i k \lambda_j} |^2$ is of interest as a probe of quantum chaos,
as is its connected part, the structure function. Plotting this average for model systems of chaotic spectra
reveals what has been termed a dip-ramp-plateau shape. Generalising earlier work of Br\'ezin and Hikami for
the Gaussian unitary ensemble, it is shown how the average in the case of the Laguerre unitary ensemble
can be reduced to an expression involving the spectral density of the Jacobi unitary ensemble. This facilitates
studying the large $N$ limit, and so quantifying the dip-ramp-plateau effect. When the parameter $a$ in the
Laguerre weight $x^a e^{-x}$ scales with $N$, quantitative agreement is found with the characteristic features
of this effect known for the  Gaussian unitary ensemble. However, for the parameter $a$ fixed, the
bulk scaled structure function is shown to have the simple functional form ${2 \over \pi} {\rm Arctan} \, k$, and
so there is no ramp-plateau transition.
\end{abstract}


\maketitle


\section{Introduction}\label{S1}

A prominent application of random matrix theory is to quantum chaos; see e.g.~the text
\cite{Ha00}. A basic postulate is that within blocks of the Hamiltonian corresponding to
good quantum numbers (e.g.~angular momentum etc.), and for large energy, the statistical
properties of the rescaled energy levels coincides with the statistical properties of the
bulk scaled eigenvalues of particular model Hamiltonians. The latter are random matrices:
$N \times N$ real (complex) Hermitian matrices $H$ formed from matrices with standard real
(complex) Gaussian entries
$X$ according to $H = {1 \over 2} ( X + X^\dagger)$ in the case that the Hamiltonian admits (does not
admit) a time reversal symmetry. In the real case, this class of random matrices is said to specify
the Gaussian orthogonal ensemble (GOE), and in the complex case the Gaussian unitary ensemble
(GUE). In a theoretical analysis, bulk scaling corresponds to first rescaling the eigenvalues
$\lambda \mapsto \pi \lambda / \sqrt{2N}$ so that in the neighbourhood of the origin the mean spacing
is unity, then taking the limit $N \to \infty$.

For spectral data, rescaling the energy levels is referred to as unfolding. 
In the case of the GOE or GUE, use can be made of the fact that
to leading order the eigenvalue density is
given by the Wigner semi-circle functional form $\rho_{(1)}(\lambda) =
(\sqrt{2N}/ \pi) (1 - \lambda^2/2 N)^{1/2}$, supported on $|\lambda| < \sqrt{2N}$;
see e.g.~\cite[eq.~(1.52)]{Fo10}. The edge eigenvalues in the data are discarded
by restricting to $|\lambda| < c \sqrt{2N}$ for a fixed $0 \ll c < 1$, thus leaving only bulk 
eigenvalues. Unfolding of these bulk eigenvalues  is carried out by rescaling $\lambda_j \mapsto
\lambda_j/ \rho_{(1)}(\lambda_j)$, so that the new density is
unity. For spectral data coming from less idealised circumstances,
where the theoretical eigenvalue density is not known, or the data is noisy, poorly resolved or even
incomplete, unfolding can no longer be precisely defined; see e.g.~\cite{Mo11}.
A fundamental question then arises: can an informative statistical quantity be found, providing
at least a qualitative indicator of quantum chaos, without unfolding?
In \cite{LLJP86} it was proposed that the structure function, also known as the spectral form factor
and first introduced into the study of quantum chaos by Berry \cite{Be85}, is well suited for this purpose.

 The spectral form factor can be viewed as  an example
of the variance of a particular linear statistic. Before specifying the
variance, let us first define the more general covariance.
Thus, with  $A = \sum_{j=1}^N a(\lambda_j)$,
$B = \sum_{j=1}^N b(\lambda_j)$ two general linear statistics, so named since $a(\lambda), b(\lambda)$ are
functions of a single eigenvalue only, the corresponding covariance is defined by
\begin{equation}\label{0.1}
{\rm Cov} \,  (A, B ) := \Big \langle (A - \langle A \rangle) (B - \langle B \rangle) \Big \rangle.
\end{equation}
Now choose
\begin{equation}\label{0.2}
 A = \sum_{j=1}^N e^{i k_1 \lambda_j}, \quad B = \sum_{j=1}^N e^{-i k_2 \lambda_j}.
\end{equation}
In the special case $k_1 = k_2 = k$, the covariance (\ref{0.1}) reduces to the variance
\begin{equation}\label{0.3}
{\rm Var} \, A := S_N(k) = \Big \langle \Big |  \sum_{j=1}^N e^{i k \lambda_j} \Big |^2  \Big \rangle -
 \Big |   \Big \langle  \sum_{j=1}^N e^{i k \lambda_j}   \Big \rangle  \Big |^2,
\end{equation}
and it is this quantity which is called the (unscaled) structure function.

The work \cite{LLJP86} identified a qualitative property of the graph of the first average on the RHS of (\ref{0.3}) --- namely the existence of a minimum value separating the small and large $k$ forms, and its  neighbourhood --- 
as an indicator of quantum chaos. There it was termed a correlation hole, and later as a dip-ramp-plateau, when it became
prominent in the course of recent studies on the scrambling of information in black holes 
\cite{C+17,CMS17} and many body quantum chaos \cite{CHLY17,TGS18,CMC19,CH19}. 
The term dip-ramp-plateau came about after the use of the GUE as a benchmark for the study of (\ref{0.3}) \cite{C+17}, and in particular the first average, where the three behaviours inherent in the name are clearly visible in the corresponding graph.
These more recent studies
also identified an analogous effect for the first term in the rewrite of the covariance
\begin{equation}\label{0.1a}
{\rm Cov} \,  (A, B ) :=  \langle A B \rangle - \langle A \rangle   \langle  B  \rangle,
\end{equation}
where $A, B$ are given by (\ref{0.2}) with
\begin{equation}\label{0.2a}
k_1 = i \Gamma + k, \qquad k_2 =- i \Gamma + k \qquad (\Gamma > 0).
\end{equation}

Some insight into the dip-ramp-plateau effect is obtained upon relating the averages (\ref{0.3}) and (\ref{0.1a})
to correlation functions. In relation to the latter, first recall that the joint eigenvalue probability density function
(PDF) for the GOE and GUE is of the form
\begin{equation}\label{0.3a}
P_N(\lambda_1,\dots,\lambda_N) = {1 \over C_N} \prod_{l=1}^N w(\lambda_l) \prod_{1 \le j < k \le N} | \lambda_k - \lambda_j |^\beta,
\end{equation}
with $w(x) = w^{(G)}(x) := e^{-\beta x^2/2}$ and $\beta = 1$ ($\beta = 2$) for the GOE (GUE). The corresponding
$k$-point correlation function $\rho_{(k)}$ is specified in terms of $P_N$ by
\begin{equation}\label{0.3b}
\rho_{(k)}(\lambda_1,\dots,\lambda_k) = {N! \over (N-k)!} \int_{-\infty}^\infty d\lambda_{k+1} \cdots   \int_{-\infty}^\infty d\lambda_{N} \, P_N(\lambda_1,\dots,\lambda_N).
\end{equation}
In the case $k=1$, this corresponds to the spectral density. The ratio $\rho_{(2)}(\lambda_1,\lambda_2)/ \rho_{(1)}(\lambda_2)$ has the
interpretation of the eigenvalue density at $\lambda_1$, given there is an eigenvalue at $\lambda_2$. Now introduce the microscopic density
\begin{equation}\label{0.3c}
n_{(1)}(\lambda) = \sum_{j=1}^N \delta (\lambda - \lambda_j),
\end{equation}
and use this to define the density-density correlation $N_{(2)}$,
\begin{equation}\label{0.3d}
N_{(2)}(\lambda,\lambda') =  {\rm Cov} \,  \Big ( n_{(1)}(\lambda), n_{(1)}(\lambda')  \Big ).
\end{equation}
The effect of the delta functions gives rise to integrals of the form (\ref{0.3b}) for $k=1$ and $k=2$, showing that
\begin{equation}\label{0.3e}
N_{(2)}(\lambda,\lambda') =  \rho_{(2)}(\lambda, \lambda') + \delta(\lambda - \lambda')  \rho_{(1)}(\lambda') -
 \rho_{(1)}(\lambda)  \rho_{(1)}(\lambda').
\end{equation}

For general linear statistics $A,B$ as defined above (\ref{0.1}), the covariance (\ref{0.1}) can be expressed in terms 
of $N_{(2)}$ as the double integral
\begin{align}\label{0.3f}
{\rm Cov} \,  (A, B ) & = \int_{-\infty}^\infty d \lambda \, a(\lambda)   \int_{-\infty}^\infty d \lambda' \, b(\lambda') \, N_{(2)}(\lambda,\lambda') \nonumber \\
& =  \int_{-\infty}^\infty d \lambda \, a(\lambda)   \int_{-\infty}^\infty d \lambda' \, b(\lambda') \, \Big ( \rho_{(2)}(\lambda, \lambda') + \delta(\lambda - \lambda')  \rho_{(1)}(\lambda') -
 \rho_{(1)}(\lambda)  \rho_{(1)}(\lambda') \Big ),
\end{align}
where the second equality follows from (\ref{0.3e}). From this second expression, separating off the term in the integrand involving the product
of densities corresponds to the form of the covariance (\ref{0.1a}), and specialising to $A,B$ given by (\ref{0.2}) with $k_1,k_2$ therein given by
(\ref{0.2a}), we deduce
\begin{multline}\label{0.3g}
 \langle A B \rangle =   \int_{-\infty}^\infty   d \lambda \, e^{(-\Gamma +i k)\lambda}  \int_{-\infty}^\infty d \lambda' \, e^{(-\Gamma -i k)\lambda'} \, \Big ( \rho_{(2)}^T(\lambda, \lambda') + \delta(\lambda - \lambda')  \rho_{(1)}(\lambda') \Big ) \\ +
 \Big |  \int_{-\infty}^\infty   e^{(-\Gamma -i k)\lambda'}  \rho_{(1)}(\lambda') \, d \lambda' \Big |^2,
 \end{multline}
where $\rho_{(2)}^T$ denotes the truncated (also known as connected) two point correlation, obtained from
$\rho_{(2)}$ by subtracting the product of the corresponding densities.

The significance of the decomposition (\ref{0.3g}) is that it distinguishes two distinct functional behaviours, both with respect to $N$,
and with respect to $k$. With respect to $N$, the first term is proportional to $N$ while the second is proportional to $N^2$. That the
first term is proportional to $N$ is a fundamental property of variances and covariances of smooth linear statistics in random matrix theory; see e.g.~\cite{PS11}.
With respect to $k$, the first term increases linearly from zero (the ramp) --- for an explanation in terms of screening in the
underlying log-gas picture, see \cite[\S 14.1]{Fo10}, or for one in terms of universality
see \cite{EY17} --- before asymptoting to a finite value (the plateau). In contrast, the second term
decreases to zero (the dip) as $k$ increases.

In the case of the GUE (indicated by the use of the subscript $(G)$), the structure function (\ref{0.3}) ${S}_N^{(G)}$ can be reduced to the single integral
via the quite striking identity
 \begin{equation}\label{bhX}
     {S}_N^{(G)}(k) =
      \int_0^k  t K_N^{(L)}(t^2/2,t^2/2) \Big |_{a=0} \, dt,
  \end{equation}
  as found by Br\'ezin and Hikami \cite{BH97} (for recent alternative derivations see \cite{Ok19,Fo20}).
  Here  $K_N^{(L)}$ denotes the correlation kernel for the Laguerre unitary ensemble (LUE), the latter
  corresponding to the eigenvalue PDF (\ref{0.3a}) with $\beta = 2$ and weight 
   \begin{equation}\label{WL}
  w(x) = w^{(L)}(x) = x^a e^{-x} \chi_{x > 0},
   \end{equation}
   where $\chi_A = 1$ for $A$ true and $\chi_A = 0$ otherwise;
  the specification of the correlation kernel is given in (\ref{1.7}) below. Thus (\ref{bhX}) is an example of
  an inter-relationship between different random matrix ensembles,  each with unitary symmetry; for others (albeit of
  a different nature) see \cite{Fo06,EL15}.
  Moreover, with the linear statistics $A,B$ given by (\ref{0.2}), as an extension of (\ref{bhX}), it
  was derived in \cite{Fo20} that
    \begin{equation}\label{H1X} 
   {\rm Cov } \,(A,B)^{(G)}
   =
   \int_0^{k_2} H^{(L)}(k_1 - k_2 +s,  s) \, ds, \quad   H^{(L)}(t_1,t_2) = {t_1 + t_2 \over 2} K_N^{(L)}(t_1^2/2, t_2^2/2) \Big |_{a=0}
  \end{equation}
(in the case that $k_1,k_2$ are given by (\ref{0.2a}) this identity was first given in \cite{Ok19}).
  
  Our aim in this paper is to seek analogues of (\ref{bhX}) and (\ref{H1X}) for the LUE --- it turns out that
  relative to the GUE the resulting structures are more complex, and we are restricted to 
  extending (\ref{bhX}).
  The quantity $S^{(L)}(k)$, in the special case $a=0$, has been the subject of attention from
  the viewpoint of numerical plots \cite{HL18} and approximate small $k$ analysis \cite{Li18}
  in the context of recent studies on the supersymmetric Sachdev-Ye-Kitaev (SYK) model
  \cite{KW17,LLXZ17,HL18}. The latter in turn is of interest both from the viewpoint of 
  information scrambling in black holes, and many body quantum chaos; see citations
  given above (\ref{0.1a}). In the case of $a$ proportional to $N$, an approximate analysis
  of $S^{(L)}(k)$ has been given in \cite{CL18} in the context of a study of the reduced density
  matrix for a chaotic many body wave function. 
  The LUE relates to supersymmetric models via the
  random chiral Hamiltonian structure
  \begin{equation}\label{HX}
  H = \begin{bmatrix} 0_{n \times n} & X \\
  X^\dagger & 0_{N \times N} \end{bmatrix}, 
  \end{equation}
  where $X$ is an $n \times N$ $(n \ge N)$ standard complex Gaussian matrix,
  with the square of the positive eigenvalues (which generally come in $\pm$ pairs) having joint distribution
  (\ref{0.3a}), weight (\ref{WL}), $a = n - N$; see \cite{Ve94}, or \cite[\S 3.1.1]{Fo10}. In relation to density matrices,
  which are positive definite matrices with unit trace, it is a fact that the eigenvalues of $X^\dagger X$ are the
  squared nonzero eigenvalues of (\ref{HX}) which gives relevance to the LUE; see e.g.~\cite{Pa93} or \cite[\S 3.3.4]{Fo10}.

  Whereas the identity (\ref{bhX})  for  $S_N^{(G)}$ involves the correlation kernel
  for the LUE (specialised to $a=0$), it turns out that the analogous expression for
  $S^{(L)}_N(k)$   involves the correlation
  kernel for the Jacobi unitary ensemble (JUE). The JUE corresponds to the eigenvalue PDF (\ref{0.3a}) with $\beta = 2$ and weight 
   \begin{equation}\label{WJ}
  w(x) = w^{(J)}(x) = x^a (1 - x)^b \chi_{1>x > 0}.
   \end{equation}
 It appears specialised to the case $b = 0$.
 
 \begin{theorem}\label{T1}
 Let $\rho_{(2)}^{T,(L)}(x,y)$ denote the truncated two-point correlation function
 for the LUE, and let $\rho_{(1)}^{(J)}(x)$ denote the eigenvalue density for the JUE.
 We have
   \begin{equation}\label{S.1} 
-   \int_{\mathbb R_+^2}  e^{i k ( x - y)} \rho_{(2)}^{T,(L)}(x,y) \, dx dy  =   
  \int_0^{1/(1 + k^2)}  \rho_{(1)}^{(J)}(x)  \Big |_{b=0} \, dx.
 \end{equation}   
 Equivalently
  \begin{equation}\label{S.2} 
 S^{(L)}_N(k)  =   
  \int_{1/(1 + k^2)}^1  \rho_{(1)}^{(J)}(x)  \Big |_{b=0}\, dx.
 \end{equation}  
  \end{theorem}
  
 An application of (\ref{S.2}) is to the calculation of the bulk scaled limit of $S^{(L)}_N(k)$.
 \begin{corollary}\label{C1}
 Define
  \begin{equation}\label{De1}
  S^{(L)}_\infty(k;\alpha) = \lim_{N \to \infty} {1 \over N}   S^{(L)}_N(k) \Big |_{a = \alpha N}.
  \end{equation} 
  Let $0 \le c < 1$ be specified by  the equation
   \begin{equation}\label{De2}  
   c = \Big ( {\alpha \over 2 + \alpha} \Big )^2,
    \end{equation}  
define 
 \begin{equation}\label{De3}  
 \rho_{(1)}^{(J), \, {\rm global}}(x) : = {1 \over \pi (1 - \sqrt{c})} {1 \over x} \sqrt{x - c \over 1 - x} \chi_{c < x < 1},
  \end{equation} 
  and specify $k_c \ge 0$ by the equation
  \begin{equation}\label{De4}   
  {1 \over 1 + k_c^2} = c = \Big ( {\alpha \over 2 + \alpha} \Big )^2.
  \end{equation}
  We have
  \begin{align}\label{De5} 
   S^{(L)}_\infty(k;\alpha) & = \int_{1/(1+k^2)}^1 
    \rho_{(1)}^{(J), \, {\rm global}}(x) \, dx \nonumber \\
  & = {2 \over \pi (1 - \sqrt{c})} \Bigg (
   - \sqrt{c} {\rm Arctan} \, \sqrt{c(1-d) \over d - c} + {\rm Arcsin} \,
   \sqrt{1 - d \over 1 - c} \Bigg ) \Bigg |_{d = 1/(1 + k^2)},
  \end{align}
  valid for $0 \le k \le k_c$, and
    \begin{equation}\label{De6}   
 S^{(L)}_\infty(k;\alpha) = 1,
 \end{equation}
 valid for $    k \ge k_c$. 
  \end{corollary}
  
  \begin{remark}\label{Re1}
  1.~The case $a$ fixed is obtained by taking $\alpha = 0$ in the above formulas.
  From (\ref{De4}) this corresponds to $k_c \to \infty$ so only the case
  (\ref{De5}) is required, which simplifies to
   \begin{equation}\label{De7} 
   S^{(L)}_\infty(k;0) = {2 \over \pi} {\rm Arctan} \, k.
    \end{equation}
    The absence of a ramp-plateau transition in the case, as distinct from
    the behaviour for $\alpha > 0$, was predicted in the work \cite{CL18} relating to random density matrices. \\
    2.~In the Appendix, prompted by a referee, an approximate analysis leading to (\ref{De7}) is presented.
    To put this in context, we recall that with
      \begin{equation}\label{De7.1}
      S_\infty^{(G)}(k) := \lim_{N \to \infty} {1 \over N} S_N^{(G)}(2 \sqrt{2N} \tau),
      \end{equation}
  Br\'ezin and  Hikami \cite{BH97} proved as a consequence of (\ref{bhX}) that
   \begin{equation}\label{De7.2x}
    S_\infty^{(G)}(k)  = \left \{ \begin{array}{ll}  {2 \over \pi} \Big ( \tau \sqrt{(1 - \tau^2)} + {\rm Arsin} \, \tau \Big ), & 0 < \tau < 1
    \\
    1, & \tau > 1. \end{array}   \right.
     \end{equation}
     In the same paper, it was shown how (\ref{De7.2x}) can be also deduced by
   approximate working based on the universal form of the bulk truncated two-point correlation
   function (see e.g.~\cite[rewrite of (7.2)]{Fo10})
    \begin{equation}\label{De7.2}
    \rho_{(2)}^{T \, {\rm bulk}}(x, y) = - {(\sin[ \pi \rho (x - y)] )^2 \over (\pi (x - y))^2 }.
    \end{equation}
    Here $\rho$ is the local eigenvalue density, the value of which depends on the choice of
    units in the bulk scaling. The idea of the referee, developed in the Appendix, is to
    use this same starting point as a mechanism which gives an explanation of the
    result (\ref{De7}).
    \end{remark}
  
In Section \ref{S2} we revise how the correlation kernel determines the correlation
functions for the LUE and JUE. For future use in the derivation of Theorem \ref{T1}, we
present differential identities for the correlation kernels in both cases, and also the evaluation of
a key definite integral involving the Laguerre polynomials. An integral evaluation relating to
the final term in (\ref{0.3g}) is derived in the first subsection of Section \ref{S3}, while the 
proof of Theorem \ref{T1}   is given in the second subsection. In Section \ref{S4}
scaled limits relevant to the dip-ramp-plateau effect are calculated, with the proof
of Corollary \ref{C1} given in the final subsection.

\section{Preliminaries}\label{S2}
Central to the study of the LUE are the Laguerre polynomials. These can be defined through the Rodrigues formula
\begin{equation}\label{1.1}
L_n^{(a)}(x) = {x^{-a} e^x \over n!} {d^n \over d x^n} \Big ( e^{-x} x^{n + a} \Big ) =
{(-1)^n \over n!} x^n + {(-1)^{n-1} (a + n)  \over (n-1)!} x^{n-1} + \cdots
\end{equation}
A convenient normalisation is to introduce a proportionality constant so that the polynomials are monic
(coefficient of leading monomial unity). Thus we define
\begin{equation}\label{1.2}
p_n^{(L)}(x) = n! (-1)^n L_n^{(a)}(x).
\end{equation}
From standard properties of the Laguerre polynomials, the corresponding orthogonality relation is
\begin{equation}\label{1.3}
\int_0^\infty x^a e^{-x} p_m^{(L)}(x)  p_n^{(L)}(x) \, dx = h_n^{(L)} \delta_{m,n}, \qquad
h_n^{(L)} = \Gamma(n+1) \Gamma(n+ a + 1).
\end{equation}

The orthogonality (\ref{1.3}) suggests introducing the orthogonal functions
\begin{equation}\label{1.4}
\psi_n^{(L)}(x) = \sqrt{w^{(L)}(x)} p_n^{(L)}(x), \qquad w^{(L)}(x) = x^a e^{-x} \chi_{x > 0}.
\end{equation}
Considering (\ref{1.4}) in squared variables, and so defining
$$
\hat{\psi}_n^{(L)}(X) = \sqrt{X} \psi_n^{(L)}(X^2),
$$
one has that $\{ \hat{\psi}_n^{(L)}(X) \}_{n=0}^\infty$ form a complete set of eigenfunctions for the Schr\"odinger
operator
$$
H^{(L)} = - {d^2 \over d X^2} + {a^2 - 1/4 \over X^2} + X^2, \qquad X > 0;
$$
see e.g.~\cite[\S 2.3]{DMS19}.
This differential operator results as a specialisation to $d=1$ of the radial part of the Schr\"odinger operator for the
$d$-dimensional harmonic oscillator --- there $a - 1/2$ relates to the quantum number for the corresponding
angular part of the same Schr\"odinger operator; see e.g.~\cite{MS96}.
The fact that $H^{(L)}$ is self adjoint with respect to the inner product
\begin{equation}\label{EpE}
\langle f, g \rangle : = \int_0^\infty f(x) g(x) \, dx
\end{equation}
gives an explanation for the orthogonality of the set of functions (\ref{1.4}).

As previously remarked, the LUE corresponds to the eigenvalue PDF (\ref{0.3}) with 
weight (\ref{WL}).
Standard theory associated with the PDFs (\ref{0.3a}) --- see
e.g.~\cite[Ch.~5]{Fo10} --- tells us that the $k$-point correlation functions
(\ref{0.3b})
have the determinantal form
\begin{equation}\label{1.6}
\rho_{(k)}(x_1,\dots, x_k)  = \det \Big [ K_N(x_j,x_l) \Big ]_{j,l=1}^k,
\end{equation}
where $K_N(x,y)$ --- referred to as the correlation kernel --- is specified by
\begin{equation}\label{1.7}
K_N(x,y) = \Big ( w(x) w(y) \Big )^{1/2} \sum_{j=0}^{N - 1} {p_j(x) p_j(y) \over h_j}.
\end{equation}
In (\ref{1.7}) $\{ p_j(x) \}$ are the set of monic orthogonal polynomials associated with the
weight function $w(x)$, normalisation $h_j$,
\begin{equation}\label{1.7a}
\int_{-\infty}^\infty w(x)  p_j(x)   p_k(x) \, dx = h_j \delta_{j,k}.
\end{equation}
Important is the explicit form of the sum in (\ref{1.7}), referred to as the Christoffel-Darboux formula
(see e.g.~\cite[Prop.~5.1.3]{Fo10})
\begin{equation}\label{1.8}
\sum_{j=0}^{N - 1} {p_j(x) p_j(y) \over h_j} = {1 \over h_{N-1}} {p_N(x) p_{N-1}(y) - p_N(y) p_{N-1}(x) \over x - y}.
\end{equation}
Hence for the LUE we have
\begin{equation}\label{KL}
K_N^{(L)}(x,y) =  {1 \over h_N^{(L)}} {\psi_{N}^{(L)}(x) \psi_{N-1}^{(L)}(y) - \psi_{N}^{(L)}(y) \psi_{N-1}^{(L)}(x) \over x - y}.
\end{equation}

Crucial to our derivation of the results of Theorem \ref{T1}  is an identity associated with the partial derivatives of (\ref{KL}) \cite{TW94c},
\cite[Proof of Prop.~5.4.2]{Fo10}.

\begin{proposition}\label{P2}
Let $\psi_n^{(L)}(x)$ be given by (\ref{1.4}) and $K_N^{(L)}(x,y)$ by (\ref{KL}). We have
\begin{equation}\label{KLa}
\Big ( x {\partial \over \partial x} +  y {\partial \over \partial y} \Big ) (xy)^{1/2} K_N^{(L)}(x,y) 
= - {(xy)^{1/2} \over 2 h_{N-1}^{(L)}} \Big ( \psi_N^{(L)}(x)  \psi_{N-1}^{(L)}(y) +
 \psi_{N-1}^{(L)}(x)  \psi_{N}^{(L)}(y) \Big ).
 \end{equation}
\end{proposition} 

\begin{proof}
We proceed as in the derivation outlined in \cite[Proof of Prop.~5.4.2]{Fo10}. The
orthogonal functions (\ref{1.4}) satisfy the matrix differential recurrence
\begin{equation}\label{KLb}
x  {d \over d x} \begin{bmatrix} \psi_n^{(L)}(x) \\  \psi_{n-1}^{(L)}(x) \end{bmatrix}  =
 \begin{bmatrix} A_{11}(x) &  A_{12}(x) \\
 A_{21}(x) & A_{22}(x) \end{bmatrix}  \begin{bmatrix} \psi_n^{(L)}(x) \\  \psi_{n-1}^{(L)}(x) \end{bmatrix},  
\end{equation}
where
\begin{equation}\label{KLc}
A_{11}(x) = - A_{22}(x) = - {1 \over 2} (x - 2n - a), \quad
A_{12}(x) = n ( a + n), \quad A_{21}(x) = 1.
\end{equation}

For general differentiable $f = f(x,y)$ we can check
$$
\Big ( x {\partial \over \partial x} +  y {\partial \over \partial y} \Big )
{(xy)^{1/2} \over x - y} f = {(xy)^{1/2} \over x - y} \Big ( x  {\partial \over \partial x} +  y {\partial \over \partial y} \Big ) f.
$$
Choosing $f =    (\psi_{N}^{(L)}(x) \psi_{N-1}^{(L)}(y) - \psi_{N}^{(L)}(y) \psi_{N-1}^{(L)}(x) )/  h_N^{(L)}$, it follows
that
\begin{align*}
& \Big ( x {\partial \over \partial x} +  y {\partial \over \partial y} \Big )
(xy)^{1/2}   K_N^{(L)}(x,y)  \\
& \quad  = {1 \over   h_N^{(L)} } {(xy)^{1/2} \over   (x - y)} \Big ( x  {\partial \over \partial x} +  y {\partial \over \partial y} \Big )
\begin{array}{cc}[ \, \psi_N^{(L)}(x) & \psi_{N-1}^{(L)}(x) \,] \\
{} & {} \end{array}
\left [ \begin{array}{cc}0 & 1 \\ -1 & 0 \end{array} \right ]
\left [ \begin{array}{l} \psi_N^{(L)}(y) \\ \psi_{N-1}^{(L)}(y) \end{array} \right ].
\end{align*}

For the partial derivatives on the RHS, use of (\ref{KLb}) shows they can be carried
out to obtain
\begin{eqnarray*}
&& \begin{array}{cc}[ \, \psi_N^{(\rm L)}(x) & \psi_{N-1}^{(\rm L)}(x) \,] \\
{} & {} \end{array}
\left [ \begin{array}{cc}
\displaystyle - {A_{21}(x) - A_{21}(y) \over x - y}
& \displaystyle  {A_{11}(x) + A_{22}(y) \over x - y}
 \\ \displaystyle  - {A_{22}(x) + A_{11}(y) \over x - y}
 & \displaystyle  {A_{12}(x) - A_{12}(y) \over x - y}
 \end{array} \right ]
\left [ \begin{array}{l} \psi_N^{(\rm L)}((y)
\\ \psi_{N-1}^{(\rm L)}((y) \end{array} \right ]
\\&& \quad = \begin{array}{cc}[ \, \psi_N^{(\rm L)}(x) &
\psi_{N-1}^{(\rm L)}(x) \,] \\
{} & {} \end{array}
\left [ \begin{array}{cc}0 &  - {1 \over 2} \\    - {1 \over 2}  & 0 \end{array} \right ]
\left [ \begin{array}{l} \psi_N^{(\rm L)}((y) \\
\psi_{N-1}(y) \end{array} \right ]
\\&& \quad = - {1 \over 2}  \Big ( \psi_N^{(\rm L)}(x) \psi_{N-1}^{(\rm L)}(y)
+ \psi_{N-1}^{(\rm L)}(x) \psi_N^{(\rm L)}(y) \Big ),
\end{eqnarray*}
and (\ref{KLa}) follows.

\end{proof}

Also of importance is the closed form evaluation of the integral
\begin{equation}\label{Jjk}
I_{jk}^{(L)}(s) := \int_0^\infty L_j^{(a)}(x) L_k^{(a)}(x) x^a e^{(s-1)x} \, dx, \quad {\rm Re} \, s < 1,
\end{equation}
which we interpret as the Laplace-Fourier transform of $ L_j^{(a)}(x) L_k^{(a)}(x) w^{(L)}(x)$.
In fact its value can be read off by specialising formulas given in
standard compendiums of integral evaluations
\cite[Entries 2.19.14.6]{PBM86}, \cite[Entry 7.414.4]{GR80},
\cite[Entry 4.11 (35)]{EMOT54}. We owe our knowledge of these references
due to them appearing in the paper \cite[\S 4]{LOS01}, which considers
further generalisations of the integrals of products of Laguerre
polynomials. In a random matrix context, (\ref{Jjk}) first appeared in the
work of Haagerup and Thorbj{\o}rnsen   \cite{HT03}, where its evaluation was
stated as a known result (with reference to the early work \cite{Ma35} also
referenced in \cite{LOS01}), and a verification type proof was given.

A companion Fourier-Laplace transform to (\ref{Jjk}) is
\begin{equation}\label{Hjk}
I_{jk}^{(G)}(s) := \int_{-\infty}^\infty e^{s x} H_j(x) H_k(x) e^{-x^2}\, dx,
\end{equation}
where $H_n(x)$ denotes the Hermite polynomial of degree $n$.
This features prominently in the derivation of (\ref{bhX}) given in
\cite{Ok19}, \cite{Fo20}, the general strategy of which underpins
our derivation of the identities of Theorem \ref{T1}. The evaluation of
(\ref{Hjk}) can be found in a number of references --- many are listed
in \cite[statement of Prop.~13]{Fo20}. The most structurally revealing make
use of generating functions. This motivates us to give a self contained
generating function approach to compute (\ref{Jjk}). For this we take as
background knowledge the generating function formulas \cite[Entry 8.975.1]{GR80}
\begin{equation}\label{L1}
\sum_{n=0}^\infty t^n L_n^{(a)}(x) = (1 - t)^{-(a+1)} e^{- t x/ (1 - t)},
\quad |t| < 1,
\end{equation}
and \cite[Eq.~5.2 (12)]{SM84}
\begin{equation}\label{J1}
\sum_{n=0}^\infty {(c)_n \over n!} \,{}_2 F_1(-n,b,c;x) t^n
 = (1 - t)^{b-c} (1 - t + x t)^{-b}, \quad |t| < 1,
\end{equation}
where
\begin{equation}\label{J2}
(c)_n  :=  {\Gamma(c + n) \over \Gamma(c)}.
\end{equation}

\begin{proposition}\label{P3}
Let ${}_2 F_1$ denote the Gauss hypergeometric function. Define $I_{jk}^{(L)}(s)$
by (\ref{Jjk}). We have
\begin{equation}\label{J3}
I_{jk}^{(L)}(s) = \Gamma(a+1) {(a+1)_j \over j!}   {(a+1)_k \over k!} (1 - s)^{-(a+1)}
\Big ( - {s \over 1 - s} \Big )^{j+k} \, {}_2 F_1(-k,-j,a + 1; 1/s^2).
\end{equation}
\end{proposition}

\begin{proof}
Use of the generating function (\ref{L1}) shows that
\begin{equation}\label{Y0}
\sum_{j,k=0}^\infty t_1^j t_2^k I_{jk}^{(L)}(s) = (1 - t_1)^{-(a+1)} (1 - t_2)^{-(a+1)} 
 \int_0^\infty e^{- t_1 x/ (1 - t_1) - t_2 x/ (1 - t_2)} x^a e^{(s-1)x} \, dx.
\end{equation}
The integal in (\ref{Y0}) reduces to the integral definition of the gamma function after a simple
change of variables. Introducing the notation
\begin{equation}\label{Y}
Y = 1 - s + t_2/(1-t_2), \qquad \tilde{Y} = (1-t_2)Y = 1 - s + s t_2,
\end{equation}
and upon some simple manipulation, this shows
\begin{equation}\label{Y1}
\sum_{j,k=0}^\infty t_1^j t_2^k I_{jk}^{(L)}(s) = \Gamma(a + 1) \tilde{Y}^{-(a+1)}
\bigg ( 1 - t_1 \Big (1 - {1 \over Y} \Big) \bigg )^{-(a+1)}.
\end{equation}

Using the binomial theorem, we read off from (\ref{Y1}) that the coefficient of $t_1^j$ is
$$
{(a + 1)_j \over j!} \Big ( 1 - {1 \over Y} \Big )^j =
{(a + 1)_j \over j!}  \bigg ( {\tilde{Y} - (1 - t_2) \over \tilde{Y} } \bigg )^j,
$$
where the equality follows from the definition (\ref{Y}). Hence we have
\begin{align}\label{Y1+}
\sum_{k=0}^\infty  t_2^k I_{jk}^{(L)}(s) & =  \Gamma(a + 1) {(a + 1)_j \over j!}  
 \tilde{Y}^{-(a+1+j)}
  \Big ( {\tilde{Y} - (1 - t_2)} \Big )^j \nonumber \\
  & =  \Gamma(a + 1) {(a + 1)_j \over j!}  (1 - s)^{-(a+1+j)} (-s)^j J(t_2;s),
 \end{align}
 where
 \begin{equation}\label{Y2}
 J(t_2;s) := (1 - \mu t_2)^{b - c} (1 - \mu t_2 + x \mu t_2)^{-b},
 \end{equation}
 with
 \begin{equation}\label{Y3} 
 b = - j, \quad c = a + 1, \quad \mu = {s \over s - 1}, \quad x = {1 \over s^2}.
 \end{equation} 
 The coefficient of $t_2^k$ in the power series expansion of (\ref{Y2}) can be read
 off from (\ref{J1}). Using this in (\ref{Y1+}) gives (\ref{J3}).
  \end{proof}
  
  The identity (\ref{S.1}) involves the correlation kernel for the JUE, with the latter in turn
  relating to the weight function (\ref{WJ}). The corresponding (monic) polynomials as
  specified by the requirement (\ref{1.7a}) are simply related to the Jacobi polynomials.
  However for present purposes it is preferable to write them in hypergeometric form
  (see e.g.~\cite{CS95})
   \begin{equation}\label{2.25} 
   p_n^{(J)}(x) = (-1)^n {(a+1)_n \over (a + b + n + 1)_n }\, {}_2 F_1 (-n,a+b+n+1,a+1;x),
   \end{equation}   
  and we read off from the same reference the explicit value of the norm
  \begin{equation}\label{2.25a} 
  h_n^{(J)} = {n!  \Gamma(a+n+1) \Gamma(b+n+1) \over  (a+b+2n+1) \Gamma(a+b+n+1)}.
  \end{equation} 
  
  We have use for a differential identity satisfied by $K_N^{(J)}(x,x)$ (note that according to
  (\ref{1.6}) this is equal to $\rho_{(1)}^{(J)}(x)$). It is a minor linear change of variables of a result in \cite[Lemma 5.6]{Le04}.
  We will give a different derivation, as a special case of a more general differential identity,
  of the type given in Proposition \ref{P2}, but now in relation to the Jacobi correlation
  kernel.
 
  \begin{proposition}\label{P4}
  We have
    \begin{equation}\label{2.26}  
    {d \over d s} \Big ( s (1 - s) K_N^{(J)}(s,s) \Big ) =
    - {(2N + a + b) \over  h_{N-1}^{(J)}} w^{(J)}(s) p_N^{(J)}(s)  p_{N-1}^{(J)}(s).
    \end{equation}
    \end{proposition}
    
    \begin{proof}
    The weight function for the Jacobi unitary ensemble supported on $(-1,1)$
    is
    $$
    \tilde{w}^{(J)}(x) = (1-x)^a(1+x)^b \chi_{-1 < x < 1}.
    $$
    Denote the corresponding correlation kernel by $\tilde{K}_N^{(J)}(x,y)$,
    the corresponding monic orthogonal polynomials by $\tilde{p}_n ^{(J)}(x)$,
    their norm by $\tilde{h}_n^{(J)}$, and corresponding orthogonal functions
    $\tilde{\psi}_n ^{(J)}(x) = \sqrt{ \tilde{w}^{(J)}(x) } \tilde{p}_n ^{(J)}(x)$.
    
    For this variant of the JUE, the Jacobi
    analogue of (\ref{KLa}) is given in \cite[2nd last displayed equation in proof of Prop.~5.4.2]{Fo10},
    \begin{multline*}
    \Big ( (1 - x^2) {\partial \over \partial x} + (1 - y^2) {\partial \over \partial y}  \Big )
    (1 - x^2)^{1/2}  (1 - y^2)^{1/2}  \tilde{K}_N^{(J)}(x,y)  \\
     = - {(2N + a + b) \over 2 \tilde{h}_{N-1}^{(J)}}
    \Big ( \tilde{\psi}_N^{(J)}(x)  \tilde{\psi}_{N-1}^{(J)}(y) +
\tilde{ \psi}_{N-1}^{(J)}(x) \tilde{ \psi}_{N}^{(J)}(y) \Big ).
  \end{multline*}
  After simple manipulation, and the linear change of variables $x = 1 - 2u$,
  $y = 1 - 2v$, this is seen to be equivalent to a differential identity for the 
  correlation kernel of the JUE as originally defined in terms of the weight
   (\ref{WJ}),
    \begin{multline*}
    \Big ( 1 - (u + v) + u (1 - u) {\partial \over \partial u} + v (1 - v) {\partial \over \partial v} \Big ) {K}_N^{(J)}(u,v) \\
    = - {(2N + a + b) \over 2 h_{N-1}^{(J)}} 
   \Big ( {\psi}_N^{(J)}(u)  {\psi}_{N-1}^{(J)}(v) +
{ \psi}_{N-1}^{(J)}(u) { \psi}_{N}^{(J)}(v) \Big ).
\end{multline*}
Taking the limit $u,v \to s$, the LHS reduces to
$$
\Big ( 1 - 2s + s (1 - s) {d \over ds} \Big ) {K}_N^{(J)}(s,s) =  {d \over ds} s (1 - s)  {K}_N^{(J)}(s,s),
$$
which is the LHS of (\ref{2.26}), and taking the same limit on the RHS, we see  the RHS of
(\ref{2.26}) results.
 \end{proof}
  
  \section{Calculation of the structure function and related averages}\label{S3}
  \subsection{Fourier-Laplace transform of the density}
  The Fourier-Laplace transform of the density appears in the expression
  (\ref{0.3g}), which in turn relates to the form of the covariance (\ref{0.1a}) as implied by (\ref{0.3f}).
 The evaluation of its derivative has been given in \cite[Th.~6.4]{HT03}. Revising
 its proof is an instructive preparation for the proof of Theorem \ref{T1}.
 
 \begin{proposition}\label{P5}
 We have
   \begin{equation}\label{Y7} 
   \int_0^\infty   t \rho_{(1)}^{(L)}(t)  e^{st} \, dt  =     {N(N+a) \over  (1 - s)^{2N+a}} \,   {}_2 F_1(-N+1-a,-N+1,2;s^2). 
  \end{equation}   
 \end{proposition}
 
 \begin{proof}
 Noting that
 $$
 \Big ( x {\partial \over \partial x} + y {\partial \over \partial y} \Big ) (xy)^{1/2}  K_N^{(L)}(tx, ty) =
 (xy)^{1/2} {d \over dt} t K_N^{(L)}(tx, ty),
 $$
 we see upon replacing $x,y$ by $xt, yt$ in (\ref{KLa}), then setting $x=y=1$, that
  \begin{equation}\label{Y4} 
  {d \over dt} t K_N^{(L)}(t,t) = - {1 \over h_{N-1}^{(L)}} \psi_N^{(L)}(t)   \psi_{N-1}^{(L)}(t) .
   \end{equation} 
  We know from (\ref{1.6}) that  $K_N^{(L)}(t,t) = \rho_{(1)}(t)$ and so after multiplying both sides
  of (\ref{Y4}) by $e^{st}$ and integrating we deduce
   \begin{align*}\label{Y5} 
   s \int_0^\infty   t K_N^{(L)}(t,t) e^{st} \, dt & = {1 \over h_{N-1}^{(L)}} \int_0^\infty \psi_N^{(L)}(t)   \psi_{N-1}^{(L)}(t)e^{st}  \, dt  \\
   & = 
   - {N! \over \Gamma(N+a)} \int_0^\infty t^a e^{-t ( 1 - s) } L_{N-1}^{(a)}(t)  L_{N}^{(a)}(t) \, dt.
   \end{align*}  
   The integrand is an example of (\ref{Jjk}), and so application of Proposition \ref{P3} gives
   $$
   \int_0^\infty   t K_N^{(L)}(t,t) e^{st} \, dt =  {(a+1)_N \over (N-1)!} {s^{2N-2} \over (1 - s)^{2N+a}} \,   {}_2 F_1(-N+1,-N,a+1;1/s^2).
  $$
  Use of the polynomial identity \cite[Eq.~(6.17)]{HT03}
   \begin{equation}\label{Y6} 
    {}_2 F_1(-j,-k,a+1;1/s^2) = {k! \over (k-j)! (a+1)_j} \, {1 \over s^{2j}} \,
 {}_2 F_1(-j-a,-j,1+k-j;s^2) ,
  \end{equation} 
  valid for $j,k$ non-negative integers with $j \le k$,
  reduces this to    (\ref{Y7}).
  \end{proof}
  
  \subsection{Proof of Theorem \ref{T1}}
  Our proof of Theorem \ref{T1} makes use of the operator
  $$
  B_{x,y} = 1 + x {\partial \over \partial x} + y {\partial \over \partial y}.
  $$
  In light of the operator identity
    $$
    (xy)^{-1/2} \Big ( x {\partial \over \partial x} +      y {\partial \over \partial y} \Big )   (xy)^{1/2}  = B_{x,y}
    $$
    we see from Proposition \ref{P2} that
    \begin{equation}\label{KLa1}
B_{x,y} K_N^{(L)}(x,y) 
= - {1 \over 2 h_{N-1}^{(L)}} \Big ( \psi_N^{(L)}(x)  \psi_{N-1}^{(L)}(y) +
 \psi_{N-1}^{(L)}(x)  \psi_{N}^{(L)}(y) \Big ).
 \end{equation}
 We observe too the skew self-adjoint property
     \begin{equation}\label{KLa2x}
     \langle f,  B_{x,y} g \rangle^{(2)} = -   \langle B_{x,y} f,   g \rangle^{(2)} , \qquad  \langle f,   g \rangle^{(2)} := \int_{\mathbb R_+^2} f(x,y) g(x,y) \, dx dy,
     \end{equation}
     as well as the identity
   \begin{equation}\label{KLa3}   
   (   B_{z_1,z_2}  - 1) e^{z_1 x + z_2 y} =    (   B_{x,y}  - 1) e^{z_1 x + z_2 y} .
    \end{equation}
    
    Application of first (\ref{KLa3}), then (\ref{KLa2x}), then a direct calculation, and finally
   (\ref{KLa1})  shows
    \begin{align}\label{3.7}
     &(   B_{z_1,z_2}  - 1)    \int_{\mathbb R_+^2}  e^{z_1 x + z_2 y} \Big ( K_N^{(L)}(x,y) \Big )^2 \, dx dy  \nonumber \\
     &\quad =    \int_{\mathbb R_+^2}   \Big ( (   B_{x,y}  - 1) e^{z_1 x + z_2 y} \Big ) \Big ( K_N^{(L)}(x,y) \Big )^2 \, dx dy   \nonumber  \\
    & \quad = - \int_{\mathbb R_+^2}  e^{z_1 x + z_2 y}  (   B_{x,y}  + 1) \Big ( K_N^{(L)}(x,y) \Big )^2 \, dx dy  \nonumber  \\
     &\quad = - 2  \int_{\mathbb R_+^2}   e^{z_1 x + z_2 y}   K_N^{(L)}(x,y)   B_{x,y}  \, K_N^{(L)}(x,y) \, dx dy  \nonumber  \\
    & \quad = {1 \over  h_{N-1}^{(L)}} \int_{\mathbb R_+^2}  e^{z_1 x + z_2 y}  K_N^{(L)}(x,y)  \Big (  \psi_N^{(L)}(x)  \psi_{N-1}^{(L)}(y) +
 \psi_{N-1}^{(L)}(x)  \psi_{N}^{(L)}(y) \Big ) \, dx dy.
     \end{align}
     Now apply the operator
     $$
  {\partial \over \partial z_1} -  {\partial \over \partial z_2}    
  $$
  to the expression in the first line of (\ref{3.7}).  In the form given in the final line of
 (\ref{3.7})  this has the effect of creating a factor $(x-y)$ inside the integrand. But from
  (\ref{KL})
  $$
  (x-y)  K_N^{(L)}(x,y) = {1 \over  h_{N-1}^{(L)}}  \Big (  \psi_N^{(L)}(x)  \psi_{N-1}^{(L)}(y) -
 \psi_{N-1}^{(L)}(x)  \psi_{N}^{(L)}(y) \Big ).
 $$
 Hence we deduce from  (\ref{3.7}) that
 \begin{align}\label{3.8}
& \Big (    {\partial \over \partial z_1} -  {\partial \over \partial z_2} \Big )  (   B_{z_1,z_2}  - 1)    \int_{\mathbb R_+^2}  e^{z_1 x + z_2 y} \Big ( K_N^{(L)}(x,y) \Big )^2 \, dx dy 
\nonumber  \\ & \qquad =
  {1 \over  (h_{N-1}^{(L)})^2} \int_{\mathbb R_+^2}  e^{z_1 x + z_2 y}  \Big (  (\psi_N^{(L)}(x)  \psi_{N-1}^{(L)}(y) )^2 -
 (\psi_{N-1}^{(L)}(x)  \psi_{N}^{(L)}(y) )^2 \Big )  \, dx dy  \nonumber \\
 & \qquad =  \Big ( {N! (N-1)! \over h_{N-1}^{(L)} } \Big )^2 \Big ( I_{N,N}^{(L)}(z_1)  I_{N-1,N-1}^{(L)}(z_2) -  I_{N,N}^{(L)}(z_2)  I_{N-1,N-1}^{(L)}(z_1)  \Big ), 
 \end{align}
 where the second equality follows from the definitions (\ref{1.4}) and (\ref{Jjk}). The significance of this expression is that according to
 Proposition \ref{P3} all terms on the RHS can be evaluated explicitly, reducing it to
    \begin{multline}\label{3.9}
     \bigg ( {(a+1)_N \over (N - 1)!} \bigg )^2 \, (1 - z_1)^{-(a+1)}  (1 - z_2)^{-(a+1)}  \Big ( {z_1 \over 1 - z_1} \Big )^{2(N-1)}  \Big ( {z_2 \over 1 - z_2} \Big )^{2(N-1)}  \\
    \bigg (   \Big ( {z_1 \over 1 - z_1} \Big )^{2}  \, {}_2 F_1(-N,-N,a+1;1/z_1^2)  \, {}_2 F_1(-N+1,-N+1,a+1;1/z_2^2) - (z_1 \leftrightarrow z_2) \bigg ).
    \end{multline}
    
    The Gaussian analogue of the equality between the first line of (\ref{3.8}) and (\ref{3.9}) is given by \cite[Equality between LHS of (3.18) and final
    expression in (3.19)]{Fo20}. Comparison between the two shows that the present Laguerre case is more complicated
    as  the first line of (\ref{3.8}) involves second order partial derivatives, whereas its Gaussian analogue only involves first
   order partial derivatives. Due to this complication, we have not been able to deduce a Laguerre analogue of 
   (\ref{H1X}). However, if we consider instead the special case of the covariance corresponding to the structure function
  (\ref{0.3}), further progress is possible.
  
  Thus set $z_1 = - z_2 = it$ in the equality between the LHS of (\ref{3.8}) and (\ref{3.9}). This gives the simplified identity
     \begin{multline}\label{3.10}
     {t \over 4} {d \over d t} \Big ( t {d \over dt} \Big ) 
    \int_{\mathbb R_+^2}  e^{i t ( x - y)} \Big ( K_N^{(L)}(x,y) \Big )^2 \, dx dy  =   \bigg ( {(a+1)_N \over (N - 1)!} \bigg )^2  \\
    \times (1 + t^2)^{-(a+1)} \Big ( {t^2 \over 1 + t^2} \Big )^{2N}
   \, {}_2 F_1(-N,-N,a+1;-1/t^2)  \, {}_2 F_1(-N+1,-N+1,a+1;-1/t^2) .
   \end{multline}
   In terms of the variable $u = 1/(1 + t^2)$ (\ref{3.10}) reads
  \begin{multline}\label{3.11}
     {d \over d u} \Big ( u (1 - u) {d \over du} \Big ) 
    \int_{\mathbb R_+^2}  e^{i \sqrt{(1-u)/u} ( x - y)} \Big ( K_N^{(L)}(x,y) \Big )^2 \, dx dy  =   \bigg ( {(a+1)_N \over (N - 1)!} \bigg )^2  \\
  \times u^a (1 - u)^{2N - 1} 
   \, {}_2 F_1(-N,-N,a+1;u/(u-1))  \, {}_2 F_1(-N+1,-N+1,a+1;u/(u-1)).
   \end{multline}  
  Recalling now the Pfaff-Kummer transformation for the Gauss hypergeometric function
  $$
   {}_2 F_1(\alpha,\beta,\gamma;z) = (1 - z )^{-\alpha} \,  {}_2 F_1(\alpha,\gamma - \beta, \gamma ;z/(z - 1)) 
 $$ 
 allows the RHS of (\ref{3.11}) to be simplified, reducing it to
  \begin{multline}\label{3.11a}
       \bigg ( {(a+1)_N \over (N - 1)!} \bigg )^2  u^a 
   \, {}_2 F_1(-N,N + a + 1,a+1;u)  \, {}_2 F_1(-N+1,N+a,a+1;u) \\
   = - { 2N + a \over h_{N-1}^{(J)} } u^a  p_N^{(J)}(u)  p_{N-1}^{(J)}(u)  \Big |_{b=0} =
     {d \over d u} \Big ( u (1 - u) K_N^{(J)}(u,u) \Big ) \Big |_{b=0},
   \end{multline}  
   where the first equality follows from (\ref{2.25}) and (\ref{2.25a}) and the second from Proposition \ref{P4}.
 Equating this to the LHS of (\ref{3.11}) and taking the indefinite integral of both sides shows
  \begin{equation}\label{KLa2}
  {d \over du} 
    \int_{\mathbb R_+^2}  e^{i \sqrt{(1-u)/u} ( x - y)} \Big ( K_N^{(L)}(x,y) \Big )^2 \, dx dy  =   
  K_N^{(J)}(u,u)  \Big |_{b=0}.
   \end{equation}
   Integrating both sides from $0$ to $s$ and setting $\sqrt{(1-s)/s} = k$ gives the sought identity
   (\ref{S.1}), upon identifying $ ( K_N^{(L)}(x,y)  )^2 = -  \rho_{(2)}^{T,(L)}(x,y) $ and
   $ K_N^{(J)}(u,u)  = \rho_{(1)}^{(J)}(u)$.
   
   The sum rules
   $$
    \int_{\mathbb R_+^2}  \delta(x-y)   \rho_{(1)}^{(L)}(y) \, dx dy = N  , \qquad
    \int_0^1  \rho_{(1)}^{(J)}(u) \, du = N,
    $$
   which are simply normalisation conditions, show that (\ref{S.2}) is equivalent to (\ref{S.1}),
   after recalling too (\ref{0.3f}).

 \section{Scaled limits}\label{S4}
 \subsection{Global scaling}
 Generally a global scaling limit in random matrix theory is when the entirety of
 the spectrum plays a role. In the LUE this takes effect when eigenvalues are scaled
 according to $\lambda_j = 4 N x_j$, the point being that in the variables $\{x_j \}$ the
 limiting support is compact. Note that this latter feature is not dependent on the
 specific choice of (positive) proportionality --- the choice of 4 is for convenience.
 There are two distinct cases: either the
 Laguerre parameter $a$ is held fixed, or the Laguerre parameter is scaled with $N$.
 In fact the former is the case $\alpha = 0$ in the second scenario, for which
  \begin{equation}\label{MP}
  \lim_{N \to \infty} 4 \rho_{(1)}^{(L)}(4Nx) \Big |_{a = N \alpha} = {2 \over \pi x}
  \sqrt{(c_+^2 - x) ( x - c_-^2)} \chi_{c_-^2 < x < c_+^2}, \quad c_\pm := {1 \over 2} \Big (\sqrt{\alpha + 1} \pm 1 \Big ),
    \end{equation}
 where this functional form is known as the Marchenko-Pastur density \cite{PS11}. Hence for a general
 linear statistic $A = \sum_{j=1}^N a(\lambda_j/4N)$
  \begin{equation}\label{MP1}
    \lim_{N \to \infty}  {1 \over N} \langle A \rangle^{(L)} = {2 \over \pi} \int_{c_-^2}^{c_+^2} { a(x) \over x}
     \sqrt{(c_+^2 - x) ( x - c_-^2)} \, dx.
   \end{equation}     
In the special case $a(x) := a_s(x) = x e^{s x}$, and with $A_s := \sum_{j=1}^N a_s(\lambda_j/4N)$,
 it follows from Proposition \ref{P5} that
(\ref{MP1}) reduces to
  \begin{equation}\label{MP2}
    \lim_{N \to \infty}  {1 \over N} \langle A_s \rangle^{(L)} =  (1 + \alpha) e^{(s/2)(1 + \alpha/2)} \, {}_0F_1(2;(1+\alpha)(s/2)^2),
 \end{equation}     
 which can also be obtained directly from (\ref{MP1}); see \cite[\S 6.6]{HT03}. 
 
 It is fundamental in random matrix theory that the variance of a smooth linear
 statistic in the global scaling limit is of $O(1)$. From  the
 definition (\ref{0.3}) of $S_N(k)$ we have that
 \begin{equation}\label{4.3a}
  {\rm Var} \, \Big ( \sum_{j=1}^N e^{i k \lambda_j/\sqrt{2} N}  \Big )^{(L)} = S_N^{(L)}(k/\sqrt{2}N) =
 \int_{1/(1 + (k/\sqrt{2}N)^2)}^1 \rho_{(1)}^{(J)}(x) \Big |_{b=0} \, dx
  \end{equation} 
 (here the factor of $\sqrt{2}$ in the global scaling is for later convenience; recall the second sentence of the
 first paragraph above), where the second equality follows from (\ref{S.2}). Recalling the symmetry of the Jacobi ensemble
 under the mappings $a \leftrightarrow b$, $x \mapsto 1 - x$ (recall (\ref{WJ})) allows us to write
\begin{multline}\label{U1}    
  \int_{1/(1 + (k/\sqrt{2}N)^2)}^1 \rho_{(1)}^{(J)}(x) \Big |_{b=0} \, dx =
   \int_0^{ (k/\sqrt{2}N)^2/(1 + (k/\sqrt{2}N)^2)} \rho_{(1)}^{(J)}(x) \Big |_{a=0} \, dx \\ =
   {1 \over 2 N^2} \int_0^{k^2/(1 + (k/\sqrt{2}N)^2)}  \rho_{(1)}^{(J)}(x/2N^2) \Big |_{a=0} \, dx ,
    \end{multline}   
   where the second equality follows by a simple change of variables.
  
   The utility of (\ref{U1}) follows from the  standard limit theorem in random matrix theory
   (see e.g.~\cite[\S 7.2.5]{Fo10}) that
  \begin{equation}\label{U2}   
  {1 \over 2 N^2}   \rho_{(1)}^{(J)}(x/2N^2) \Big |_{a=0}  = \rho_{(1)}^{\rm hard}(x) \Big |_{a=0}  =
  {1 \over 4} \Big ( (J_0(x^{1/2}))^2 +  (J_1(x^{1/2}))^2 \Big ),
  \end{equation}    
  where $J_n(v)$ denotes the Bessel function of order $n$ and $ \rho_{(1)}^{\rm hard}(x)$ denotes the scaled hard
  edge state with unitary symmetry \cite{Fo93a}, with a remainder term that can readily be checked to be
  uniform for $x$ on a compact set of the positive half line. Hence
  \begin{multline}\label{U3}
   S_\infty^{(L), {\rm global} }(k) :=   \lim_{N \to \infty} 
  S_N^{(L)}(k/\sqrt{2}N) =  \int_0^{k^2}   \rho_{(1)}^{\rm hard}(x) \Big |_{a=0}   \, dx \\
 = {1 \over 4} \int_0^{k^2} \Big (  (J_0(x^{1/2}))^2 +  (J_1(x^{1/2}))^2 \Big ) \, dx.
   \end{multline}
   We remark that since for $x \to \infty$, $  \rho_{(1)}^{\rm hard}(x) \sim 1/(2 \pi x^{1/2})$ (this holds
   independent of the parameter $a$; see \cite[Eq.~(7.74)]{Fo10}), for $k \to \infty$
   \begin{equation}\label{U4} 
     S_\infty^{(L), {\rm global} }(k) 
   \sim {k \over  \pi} . 
    \end{equation}
   This corresponds to the `ramp' in the dip-ramp-plateau effect discussed in the Introduction.
   We remark too that the global scaling limit of the structure function for the GUE as implied by
   (\ref{bhX}) is also given by the same functional form (\ref{U3}) (note that this is dependent on the
   precise choice of the proportionality in the global scaling --- more generally this statement would
   hold after appropriately identifying $k$) \cite{Fo20}. In the latter reference it is noted that the
   integral in (\ref{U3}) can be evaluated explicity; see \cite[Eq.~(3.28)]{Fo20}.
   
   In addition to the variance of a smooth linear statistic in the global scaling limit being $O(1)$,
   another generic feature is that their limiting distribution satisfies a central limit theorem;
   see \cite{PS11}. Recently the question of the rate of convergence to the central limit theorem
   has attracted attention from a number of different viewpoints \cite{LLW19,BB19,JKM20}.
   The formula (\ref{4.3a}) allows the convergence rate question to be addressed for the variance of
   the specific linear statistic relating to the structure function in the LUE. According to (\ref{U1}),
   this is determined in turn by the rate of convergence of the hard edge scaled density for the
   JUE. On this, we have the recent large $N$ expansion \cite[Prop.~2]{MMM19} (see also the
   related work \cite{FT19})
    \begin{equation}\label{U5} 
    \rho_{(1)}^{(J)}(x/2N^2) \Big |_{a=0} = \rho_{(1)}^{\rm hard}(x) \Big |_{a=0} +
    {b \over N} x {d \over dx}  \rho_{(1)}^{\rm hard}(x) \Big |_{a=0}  + O \Big ( {1 \over N^2} \Big ),
    \end{equation}
   telling us that the rate is $O(1/N)$.

   \subsection{Bulk scaling of the linear statistic $\sum_{j=1}^N e^{i k \lambda_j}$}
   Bulk scaling refers to using a linear change of variables so that the eigenvalues away from the
   edges have nearest neighbour of order unity for $N$ large. For the Laguerre ensemble, the
   support of the eigenvalue density is an interval of length proportional to $N$, in both the cases
   of $a$ fixed or proportional to $N$
   
   Before considering the corresponding limiting form of $S_N^{(L)}$, in view of the interest in
   (\ref{0.3g}) for $A,B$ given by (\ref{0.2}) with $k_1 = k_2 = k$, we first make some
   remarks in relation to the average of the linear statistic $\sum_{j=1}^N e^{i k \lambda_j}$.
   For $N,k$ large, it follows from (\ref{MP1}) that
    \begin{equation}\label{MP+}
  \Big \langle \sum_{j=1}^N e^{i k \lambda_j}  \Big   \rangle^{(L)} \sim {2 N \over \pi} \int_{c_-^2}^{c_+^2} { e^{i 4 N k x} \over x}
     \sqrt{(c_+^2 - x) ( x - c_-^2)} \, dx.
   \end{equation}   
   There are two distinct behaviours, depending on $\alpha = 0$ (and thus $c_- = 0$) or $\alpha > 0$.
   In the former, expanding the integrand in the neighbourhood of $x=0$ and changing variables shows
     \begin{equation}\label{Ma} 
    \Big \langle \sum_{j=1}^N e^{i k \lambda_j}  \Big   \rangle^{(L)} \sim  \sqrt{N \over i \pi k}.
 \end{equation}  
 As noted in \cite{HL18}, the absolute value squared of (\ref{Ma}) being of order $N$, and its slow
 (relative to the Gaussian case \cite{C+17,Fo20}) $O(1/k)$ `dip' obscures the `ramp' in the
 dip-ramp-plateau effect. In contrast, for $\alpha > 0$, expanding the integrand in the neighbourhoods
 of both endpoints $c_+^2> c_-^2>0$ shows that for some $s_\pm(\alpha)$ independent of $k,N$
   \begin{equation}\label{Mb} 
    \Big \langle \sum_{j=1}^N e^{i k \lambda_j}  \Big   \rangle^{(L)} \sim   {1 \over N^{1/2}} {1 \over k^{3/2}}
    \Big ( s_+(\alpha) e^{i c_+^2 4 N k} + s_-(\alpha) e^{i c_-^2 4 N k}  \Big ),
 \end{equation}  
 which exhibits the same rate of decay in both $N$ and $k$ as the Gaussian case, indicating that the
 dip and ramp are distinct effects (see also \cite[Section IV.B with the identification $k = \tau/ N$]{CL18}).
 
 \subsection{Bulk scaling of the structure function $S_N^{(L)}$ and proof of Corollary \ref{C1}}
 The limiting form of $S_N^{(L)}$ is easy to compute from (\ref{S.2}). The latter reduces the task to computing
 the global limiting form of the density of the JUE --- it is a global scaling since there are of order $N$
 eigenvalues in the interval $(c,1)$ for any $0<c < 1$. This is known in random matrix theory from a
 result of Wachter  \cite{Wa78},
   \begin{equation}\label{Wa}  
  \lim_{N \to \infty} {1 \over N}  \rho_{(1)}^{\rm JUE}(x) \Big |_{b=0, a = \alpha N} =
  {1 \over \pi (1 - \sqrt{c})} {1 \over x} \sqrt{x - c \over 1 - x} \chi_{c < x < 1} , \quad   {1 \over (1 - \sqrt{c})}  = 1 + \alpha/2.
   \end{equation}   
   The statement of Corollary \ref{C1} now follows, where the integral in (\ref{De5}) has been evaluated
   with the help of  computer algebra.
   
   The case of fixed $a$ corresponds to the case $\alpha = 0$ ($c = 0$) in this formula and so
   \begin{equation}\label{Wa1}  
  \lim_{N \to \infty} {1 \over N}  \rho_{(1)}^{(J)}(x) \Big |_{b=0,  \: a \: {\rm fixed}} =
 {1 \over \pi}  {1 \over  \sqrt{x (1 - x)}} \chi_{0 < x < 1}.
   \end{equation}  
  Use of this in (\ref{S.2}) validates (\ref{De7}).
   
   Contrary to the results of an approximate analysis \cite[Eq.~(3.15), Fig.~3]{HL18},
   \cite[Eq.~(4.17), Figure 3]{Li18}, our exact result (\ref{De7}) shows that for the LUE with
   $a$ fixed there is no transition from a ramp to plateau in the graphical shape of
   $S_\infty^{(L)}(k;0)$. The exact result exhibits the limiting forms
     \begin{align}\label{Wa2}   
 S_\infty^{(L)}(k;0)  & \mathop{\sim}\limits_{k \to 0^+}  {2k \over \pi} - {2k^3 \over 3 \pi } + O(k^5)  \nonumber \\
  S_\infty^{(L)}(k;0)  & \mathop{\sim}\limits_{k \to \infty}  1 - {2 \over \pi k} + {2 \over 3 \pi k^3} + O  (k^{-5}  ),
  \end{align}
  and $S_\infty^{(L)}(k;0)$ is real analytic for $k > 0$.
   
    In contrast to the behaviour of $S_\infty^{(L)}(k;0)$,  (\ref{De5}) and (\ref{De6}) show that with $a = \alpha N$ there is a  transition to a plateau
    $S_\infty^{(L)}(k;\alpha) = 1$, occurring at the value of $k$ specified by (\ref{De4}). Like in the
    Gaussian case \cite{BH97} the ramp portion of the graph is curved, although the leading
    small $k$ form is linear
    \begin{equation}\label{Wa3}  
     S_\infty^{(L)}(k;\alpha)   \mathop{\sim}\limits_{k \to 0^+}  {2 \sqrt{1 + \alpha}  \over \pi} k  + O(k^2) .
  \end{equation}  
  Graphical plots indicate that for $0 < k < k_c$, $S_\infty^{(L)}(k;\alpha)$ is concave, with curvature increasing
  as $\alpha$ decreases.
  As $k \to k_c^-$, use of the first of the expressions in     (\ref{De5}) shows
   \begin{equation}\label{Wa4}  
    S_\infty^{(L)}(k;\alpha)   \mathop{\sim}\limits_{k \to  k_c^-}  1-  {2 \over 3 \pi (1 - \sqrt{c})} {1 \over c \sqrt{1 - c}} \bigg ( \bigg (
    {k_c^2 - k^2 \over (1 + k^2)(1 + k_c^2)} \bigg )^{3/2} + O \bigg (
    {k_c^2 - k^2 \over (1 + k^2)(1 + k_c^2)} \bigg )^{5/2} \bigg ).
   \end{equation}  
   Hence both the function value (which is equal to 1), and the value of its first derivative    
   (which is equal to $0$) agree at the transition to the plateau.

  \subsection*{Acknowledgements}
	This research is part of the program of study supported
	by the Australian Research Council Centre of Excellence ACEMS
	and the Discovery Project grant DP210102887.
	The efforts of the referees in improving the paper is acknowledged.

\appendix
\section*{Appendix}\label{A1}
\renewcommand{\thesection}{A} 
\setcounter{equation}{0}
In Remark \ref{Re1}.2 it was noted that a referee outlined to the author an
approximate analysis which reproduces (\ref{De7}). This analysis is based
on the universal form (\ref{De7.2}), or equivalently its Fourier transform
\begin{equation}\label{A.1}
\int_{-\infty}^\infty {\sin^2(\pi \rho w) \over w^2} e^{i w \tau} \, dw =
\left \{  \begin{array}{ll} \pi^2 \Big ( \rho - {\tau /( 2 \pi)} \Big ), &  \tau/(2 \pi) < \rho \\
0, &  \tau/(2 \pi) > \rho, \end{array} \right.
\end{equation}
where it is assumed $\tau > 0$.
In the same remark we commented that this idea can be also found in the original
paper of Br\'ezin and Hikami \cite{BH97}, where it was used to anticipate the findings of
their exact analysis giving the functional form (\ref{De7.2x}) for the GUE; see also
\cite[\S 3]{Ok19}.

To present the argument, first note that combining (\ref{0.3}) with the first line of
(\ref{0.3g}) in the case $\Gamma = 0$ shows
\begin{equation}\label{A.2}
S_N(k) = \int_{-\infty}^\infty d \lambda \, e^{i k \lambda} \int_{-\infty}^\infty d \lambda' \, e^{i k \lambda'}
\Big ( \rho_{(2)}^T(\lambda, \lambda') + \delta(\lambda - \lambda') \rho_{(1)}(\lambda') \Big ).
\end{equation}
Next change variables $\lambda, \lambda' \mapsto 4N \lambda, 4N \lambda'$ so that the density
for large $N$ has the leading form
\begin{equation}\label{A.3}
4 \rho_{(1)}(4 N \lambda) \sim N \rho^{\rm MP}(\lambda), \qquad
\rho^{\rm MP}(\lambda) := {2 \over \pi} \sqrt{(1-\lambda)/\lambda}  \, \chi_{0 < \lambda < 1}
\end{equation}
(this is (\ref{MP}) with $\alpha = 0$). 
Following  \cite{BH97}, or the independent working of the referee, the
key hypothesis is the approximation
\begin{equation}\label{A.3a}
4^2 \rho_{(2)}^T(4 N \lambda, 4 N \lambda') \approx -
{ ( \sin [ \pi N \rho^{\rm MP}((\lambda + \lambda')/2) ( \lambda - \lambda')] )^2 \over
(\pi ( \lambda - \lambda'))^2 } \chi_{0 < \lambda, \lambda' < 1}
\end{equation}
which is based on the universal bulk scaling form (\ref{De7.2}).

Noting that the double integral of the second term in (\ref{A.2}) equals $N$ as a 
normalisation, and changing variables
$$
w = N (\lambda - \lambda'), \qquad u =  (\lambda + \lambda')/2
$$
then gives as a large $N$ approximation
\begin{equation}\label{A.4}
{1 \over N} S_N^{(L)}(k) \approx 1 - {1 \over \pi^2} \int_{-\infty}^\infty dw \int_0^1 du \,
{\sin^2 (\pi \rho^{\rm MP}(u) w) e^{4 i w k} \over w^2}.
\end{equation}
Changing the order of integration, and using (\ref{A.1}) with
$\tau = 4 k > 0$ then shows
\begin{equation}\label{A.4x}	 
{1 \over N} S_N^{(L)}(k) \approx 1 -  \int_0^{u^*} \Big ( \rho^{\rm MP}(u)  - {2k \over \pi} \Big ) \, du,
\end{equation}
where $u^* = u^*(k)$ is such that
\begin{equation}\label{A.4a}
{\pi \over 2} \rho^{\rm MP}(u^*) = k,
\end{equation}
and thus from (\ref{A.3}) has the explicit value
\begin{equation}\label{A.4b}
u^* = {1 \over 1 + k^2}.
\end{equation}

Differentiating (\ref{A.4x}) with respect to $k$, and making use of
(\ref{A.4a}) and (\ref{A.4b}) shows
\begin{equation}\label{A.4c}
{d \over dk} \Big ( {1 \over N} S_N(k) \Big ) \approx {2 \over \pi} u^* = {2 \over \pi} {1 \over 1 + k^2},
\end{equation}
which we see is in precise agreement with (\ref{De7}).
As emphasised by the referee, a further prediction of this working is that for a random matrix
ensemble in the unitary symmetry class, and thus possessing a bulk scaled two-point function
(\ref{De7.2}), the ramp-plateau transition will be absent whenever the corresponding spectral
density is unbounded. Thus the distinction of the behaviours (\ref{De7}) when the limiting spectral
density is given by (\ref{A.3}), and (\ref{De5}), (\ref{De6}) when the limiting
spectral
density is given by (\ref{MP}).


\begin{thebibliography}{10}

\bibitem{BB19}
S.~Berezin and A.I.~Bufetov,
 \emph{On the rate of convergence in the central limit theorem for linear statistics of
 Gaussian, Laguerre and Jacobi ensembles}, arXiv:1904.09685.
 
 \bibitem{Be85}
 M.V.~Berry, \emph{Semiclassical theory of spectral rigidity}, Proc.~R.~Soc.~Lond. A \textbf{400}, 229--251.


\bibitem{BH97}
E.~Br\'ezin and S. Hikami, \emph{Spectral form factor in a random matrix theory}, Phys. Rev. E
\textbf{55} (1997), 4067--4083.  



\bibitem{CMS17}
A. del Campo, J. Molina-Vilaplana and J. Sonner, \emph{Scrambling the spectral form factor:
unitarity constraints and exact results}, Phys. Rev. D \textbf{95} (2017), 126008. 

\bibitem{CL18} X.~Chen and A.W.W.~Ludwig, \emph{Universal spectral correlations in the chaotic
wave function, and the development of quantum chaos}, Phys.~Rev.~B \textbf{98} (2018), 064309.

\bibitem{CS95}
M.-P.~Chen and H.M.~Srivastava, \emph{Orthogonality relations and generating functions for
Jacobi polynomials and related hypergeometric functions},
Applied Math.~Comp. \textbf{68} (1995), 153--188.

\bibitem{CMC19}
A.~Chenu, J.~Molina-Vilaplana and A.~del Campo, \emph{Work statistics, Loschmidt echo
and information scrambling in chaotic quantum systems}, Quantum \textbf{3}, (2019) 127 

\bibitem{C+17} J.S. Cotler, G. Gur-Ari, M. Hanada, J. Polchinski, P. Saad, S.H. Shenker, D. Stanford,
A. Streicher and M. Tezuka, \emph{Black Holes and Random Matrices}, JHEP \textbf{1705} (2017), 118;
Erratum: [JHEP \textbf{1809} (2018), 002]

\bibitem{CH19} J.S. Cotler and N.~Hunter-Jones, \emph{Spectral decoupling in many-body quantum
chaos}, arXiv:1911.02026

\bibitem{CHLY17}
J.S. Cotler, N. Hunter-Jones, J. Liu and B. Yoshida, \emph{Chaos, Complexity, and Random
Matrices}  JHEP \textbf{1711} (2017), 048 

\bibitem{DMS19}
D. S. Dean, P. Le Doussal, S. N. Majumdar, and G. Schehr, \emph{Noninteracting fermions in a trap and random
matrix theory}, J. Phys. A \textbf{52} (2019), 144006.

\bibitem{EL15}
A.~Edelman and M.~La Croix, \emph{The singular values of the GUE (less is more)},
Random Matrices: Th.~Appl. \textbf{4} (2015) 1550021.


\bibitem{EMOT54}
A.~Erd\'elyi, W.~Magnus, F.~Oberhettinger and F.G.~Tricomi, \emph{Tables of Integral
Transforms}, Vol.~I, McGraw-Hill, New York, Toronto, and London, 1954.

\bibitem{EY17}
L. Erd\"os and H.-T. Yau, \emph{A dynamical approach to random matrix theory},
Courant Lecture Notes in Mathematics, vol. 28, Amer. Math. Soc. Providence, 2017.


\bibitem{Fo93a}
P.J. Forrester, \emph{The spectrum edge of random matrix ensembles}, Nucl.
  Phys. B \textbf{402} (1993), 709--728.
  
 \bibitem{Fo06} 
  P.J. Forrester,
  \emph{Evenness symmetry and inter-relationships between gap probabilities in random matrix theory},
  Forum Math. \textbf{18} (2006), 711--743.

 \bibitem{Fo10}
P.J.~Forrester, \emph{Log-gases and random matrices}, Princeton University Press,
  Princeton, NJ, 2010.

 \bibitem{Fo20}
 P.J.~Forrester, \emph{Differential identities for the structure function of some random matrix ensembles},
 J. Stat. Phys. \textbf{183} (2021), 1--28.
 
 \bibitem{FT19} P.J.~Forrester and  A. K.~Trinh,
\emph{Finite size corrections at the hard edge for the Laguerre $\beta$ ensemble},
Stud. Appl. Math. \textbf{143} (2019), 315--336.
 
 \bibitem{GR80}
I.S. Gradshteyn and I.M. Ryzhik, \emph{Table of integrals, series, and
  products}, 4th ed., Academic Press, New York, 1980.
  
   \bibitem{HT03}
U.~Haagerup and S.~Thorbj{\o}rnsen, \emph{Random matrices with complex {G}aussian
	entries}, Expo. Math. \textbf{21} (2003), 293--337.

  

 \bibitem{Ha00}
F.~Haake, \emph{Quantum signatures of chaos}, 2nd ed., Springer, Berlin, 2000.

 \bibitem{HL18}
N. Hunter-Jones, J. Liu,  \emph{Chaos and random matrices in supersymmetric SYK},
JHEP \textbf{2018} (2018), 202.

 \bibitem{JKM20}
 B.~Jonnadula, J.P.~Keating and F.~Mezzadri,
  \emph{Symmetry function theory and unitary invariant ensembles},
  arXiv:2003.02620.

 \bibitem{KW17}
T. Kanazawa and T. Wettig, \emph{Complete random matrix classication of SYK models
with $\mathcal N = 0, 1$ and 2 supersymmetry}, JHEP \textbf{2017} (2017) 050.

 \bibitem{LLW19}
 G.~Lambert, M.~Ledoux and C.~Webb, \emph{Quantitative normal approximation of linear
 statistics of $\beta$-ensembles}, 
 Ann.~Prob.~\textbf{47} (2019), 2619--2685.
 

\bibitem{Le04}
M.~Ledoux,
\emph{Differential operators and spectral distributions of invariant ensembles
	from the classical orthogonal polynomials. The continuous case},
Electron. J. Probab. \textbf{9} (2004), 177--208.

\bibitem{LOS01}
P.-A.~Lee, S.-H.~Ong and H.M.~Srivastava, \emph{Some integrals of the products
of Laguerre polynomials},
Int. J. Comp. Math. \textbf{78} (2001), 303--321.

 \bibitem{LLJP86} 
 L. Leviandier, M. Lombardi, R. Jost and J. P. Pique,
 \emph{Fourier Transform: A Tool to Measure
Statistical Level Properties in Very Complex Spectra}, Phys. Rev. Lett. \textbf{56} (1986), 2449 .

\bibitem{LLXZ17}
T.~Li, J.~Liu, Y.~Xin and Y.~Zhou,
\emph{Supersymmetric SKY model and random matrix theory},
JHEP \textbf{2017} (2017) 111

\bibitem{Li18}
J. Liu, \emph{Spectral form factors and late time quantum chaos}, Phys. Rev. D \textbf{98} (2018), 
086026.


\bibitem{Ma35}
K.~Mayr, \emph{ Integraleigenschaften der Hermiteschen und Laguerreschen Polynome},
Math. Zeitschr. \textbf{39} (1935), 597 -604.





\bibitem{Mo11}
I.O.~Morales,  E.~Landa,  P. Str\'ansk\'y,  A.~Frank,  \emph{Improved  unfolding  by detrending  of  statistical  fluctuations  in  quantum  spectra},
  Phys  Rev  E \textbf{84} (2011), 016203
  
  \bibitem{MMM19} L.~Moreno-Pozas, D.~Morales-Jimenez and M.R.~McKay,
 \emph{Extreme eigenvalue distributions of Jacobi ensembles: new exact
 representations, asymptotics and finite size corrections}, Nucl.~Phys.~B {\bf 947} (2019), 114724.

\bibitem{MS96}
 M.~Moshinsky, Y.F.~Smirnov, \emph{The  Harmonic  oscillator  in  modern  physics}, 
 (Contemporary Concepts in Physics Volume 9), Harwood Academic Publishers,  Amsterdam (1996)

 \bibitem{Ok19}
  K. Okuyama, \emph{Spectral form factor and semi-circle law in the time direction},
  JHEP \textbf{2019} (2019), 161. 
  
  \bibitem{Pa93}
  D.N.~Page, \emph{Average entropy of a
subsystem}, Phys. Rev. Lett. \textbf{71} (1993), 1291--1294.

 
   \bibitem{PS11}
L.~Pastur and M.~Shcherbina, \emph{Eigenvalue distribution of large random
  matrices}, American Mathematical Society, Providence, RI, 2011.
  
  \bibitem{PBM86}
  A.P.~Prudnikov, Yu.A. Brychkov and O.L.~Marichev, \emph{Integrals and Series, Vol 2:
Special Functions}, 
Gordon and Breach, New York, 1986.

\bibitem{SM84}
H.M.~Srivastava and H.L.~Manocha, \emph{A treatise on generating functions}, John Wiley and Sons/ Ellis
Horwood, Chichester, 1984.
 

\bibitem{TGS18}
E.J.~Torres-Herrera, A.M.~Garc\'ia-Garc\'ia,  and L.F.~Santos, 
\emph{Generic dynamical features of quenched interacting quantum systems: 
Survival probability, density  imbalance,  and  out-of-time-ordered  correlator}, Phys. Rev. B
\textbf{97} (2018), 060303.

\bibitem{TW94c}
C.A.~Tracy and H.~Widom, \emph{Fredholm determinants, differential equations and matrix
  models}, Commun. Math. Phys. \textbf{163} (1994), 33--72.
  
  \bibitem{Ve94}
J.J.M. Verbaarschot, \emph{The spectrum of the {Dirac} operator near zero
  virtuality for $n_c=2$ and chiral random matrix theory}, Nucl. Phys. B
  \textbf{426} (1994), 559--574.
  
  \bibitem{Wa78}
K.W. Wachter, \emph{The strong limits of random matrix spectra for sample
  matrices of independent elements}, Annal. Prob. \textbf{6} (1978), 1--18.
 
 
  \end{thebibliography}
\nopagebreak

\providecommand{\bysame}{\leavevmode\hbox to3em{\hrulefill}\thinspace}
\providecommand{\MR}{\relax\ifhmode\unskip\space\fi MR }
\providecommand{\MRhref}[2]{%
  \href{http://www.ams.org/mathscinet-getitem?mr=#1}{#2}
}
\providecommand{\href}[2]{#2}

   \end{document}